\documentclass[aps,pra,twocolumn,groupedaddress]{revtex4-1}
\usepackage{amsfonts,amsthm}
\usepackage{url}
\usepackage{amssymb}
\usepackage{graphicx}
\usepackage{color}
\usepackage{amsmath}
\newcommand{\bip}[2]{\ensuremath{K_{#1,#2}}}
\newcommand{\incbip}[2]{\ensuremath{\widetilde{K}_{#1,#2}}}
\newcommand{\g}[1]{\ensuremath{\mathcal{#1}}}
\newcommand{\Kn}[1]{\ensuremath{K_{#1}}}
\newcommand{\Ch}[3]{\ensuremath{\text{Chimera}(#1,#2,#3)}}

\newtheorem{example}{Example}
\newtheorem{theorem}{Theorem}
\newtheorem{lemma}[theorem]{Lemma}
\newtheorem{corollary}[theorem]{Corollary}

\begin{document}

\title{Identifying the minor set cover of dense connected bipartite graphs via random matching edge sets}

\author{Kathleen E. Hamilton}
\email{hamiltonke@ornl.gov}
\author{Travis S. Humble}

\affiliation{Quantum Computing Institute, Oak Ridge National Laboratory, Tennessee, 37821, USA}

\begin{abstract}
Using quantum annealing to solve an optimization problem requires minor embeddings of a logic graph into a known hardware graph. In an effort to reduce the complexity of the minor embedding problem, we introduce the minor set cover (MSC) of a known graph \g{G}: a subset of graph minors which contain any remaining minor of the graph as a subgraph. Any graph that can be embedded into \g{G} will be embeddable into a member of the MSC. Focusing on embedding into the hardware graph of commercially available quantum annealers, we establish the MSC for a particular known virtual hardware, which is a complete bipartite graph. We show that the complete bipartite graph \bip{N}{N} has a MSC of $N$ minors, from which \Kn{N+1} is identified as the largest clique minor of \bip{N}{N}. The case of determining the largest clique minor of hardware with faults is briefly discussed but remains an open question.

\end{abstract}
\keywords{minor embedding, adiabatic quantum computing, quantum annealing, clique minor, graph theory}

\maketitle
\section{Introduction}
Adiabatic quantum computation uses a continuous-time process to evolve the state of a quantum register \cite{farhi2001quantum}. Whereas the register elements are represented by quantum physical subsystems that can store qubits of information, the continuous-time evolution depends explicitly on a Hamiltonian that defines the interactions between register elements \cite{kaminsky2004scalable}. An ideal Hamiltonian may allow for arbitrary interactions between elements, but physical and technological limitations often prevent fabrication of arbitrary interactions or forms of connections in actual devices. A prominent example is found in the quantum annealer developed by D-Wave Systems, Inc. \cite{Johnson2011}, which uses a well-defined hardware connectivity graph called the Chimera lattice and implements problems which can be described using the Ising Hamiltonian with two-body interactions. 

The problem of expressing an arbitrary Hamiltonian in the presence of limited connections poses a prominent concern for quantum annealing applications \cite{humble2014integrated}. Presently we consider the input to this programming process to be a well-defined logical Ising Hamiltonian.  The logical Ising model is known to capture a broad class of different problems and may also be presented in quadratic unconstrained binary optimization (QUBO) form \cite{lucas2014ising}.  Choi originally formulated the process of programming a logical problem as graph minor embedding, in which a graph representing the dependency of the input Hamiltonian is mapped into the targeted hardware graph \cite{choi2008minor,choi2011minor}.
In general, graph minor embedding requires each logical vertex to be mapped into a connected subtree of the hardware graph, and Choi's TRIAD algorithm yielded a deterministic method for embedding a complete graph into the Chimera lattice. There are many approaches for determining if a graph can be embedded into another graph, including: identifying useful graph qualities\cite{grigoriev2007algorithms,coury2007embedding,fomin2014preprocessing}, establishing a set of known embeddings as a ``lookup table'' \cite{klymko2014maximal,zaribafiyan2016systematic} or other heuristic methods \cite{cai2014minor,boothby2016fast}.  

More recently, Cai, Maccready, and Roy (CMR) have presented a randomized algorithm for generating an embedding \cite{cai2014minor}. Their approach is based on employing Djikstra's algorithm to find the shortest-path between randomly mapped logical vertices. The CMR algorithm has proven useful for embedding arbitrary input graphs in current hardware because it can find smaller embeddings than using either the TRIAD algorithm \cite{choi2011minor} or maximal minor embedding \cite{klymko2014maximal}. However, the method is not guaranteed to succeed and has a worst case complexity that scales as $O(n^9)$ with the input graph order $n$ (though average case behavior appears to be $O(n^3)$). The CMR embedding algorithm represents a significant portion of the time needed for a quantum annealing workflow, and for even modest problem sizes it can far exceed the time required for executing a quantum annealing schedule \cite{Humble2016APDCM}.  Problem instances represented by large but incompletely connected input graphs must use embeddings that are both resource efficient and time efficient in order to ensure fast and correct solutions. Examples include dynamic job scheduling \cite{venturelli2015scheduling} and route planning \cite{rieffel2015case}, as well as time-dependent fault-detection \cite{perdomo2015quantum}.
Alleviating the classical processing bottleneck while retaining the resource efficiency of the CMR algorithm is therefore an important problem for solving optimization problems with quantum annealing and integrating these quantum processing units into future computing systems \cite{Britt2015}.

In this contribution, we present a quasi-deterministic method for graph minor embedding that takes advantage of a virtual hardware abstraction. 
Our method builds upon two recent embedding concepts: ``maximal minor embedding'' developed in \cite{klymko2014maximal}, which characterizes finding the most efficient embedding with respect to the minimal number of hardware nodes used in the embedded graph; and recent ideas developed by Goodrich and collaborators\footnote{\label{note1}Work presented at the Society for Industrial and Applied Mathematics Workshop on Network Science 2016, manuscript in progress.} that uses the complete bipartite (biclique) minor of the Chimera graph as a virtual hardware (as shown in Fig.~\ref{fig:K13_embedding}).  The full embedding of a problem graph in the hardware graph is found by first embedding into a chosen virtual hardware or one of its minors. Our choice of a bipartite virtual representation for the hardware is motivated in part by the simplicity of the structure as well as its balance between size and order of the virtual representation, and additionally for its connection to associative memory recall and other variants of machine learning applications \cite{seddiqi2014adiabatic}. We note that alternative virtual representations are equally valid, e.g., a square grid. 

Klymko et al.\cite{klymko2014maximal} established tight bounds on the largest complete graph that can be minor embedded in a Chimera graph as well as demonstrating methods for embedding into faulty hardware graphs and introduced the concept of ``maximal minor'' embedding. We rename the set of ``maximal minors'' as the \textit{minor set cover} (MSC): this is the set of minors for a given graph \g{G} where any subgraph or minor of \g{G} will either be a member of the MSC or is a subgraph contained in one of the members. The MSC of the biclique virtual hardware is a finite set of embeddable graph minors which can be precomputed without reference to the input problem, and can act as a lookup table. This reframes the problem of graph minor embedding as a subgraph isomorphism search: an input graph is compared against each key, if it is found to be contained in a given key, the problem graph is now embedded into a member of the MSC, and then an embedding in the hardware is found. It has been proposed previously that reducing the problem complexity of graph embedding from graph homomorphism to subgraph isomorphism can lead to substantial speedup in processing time \cite{HELL199092,cygan2016tight}.

The question of whether a graph is minor-embeddable has been explored through many different approaches, such as: forbidden minors and extremal graph theory.  Forbidden minors are minors which a class of graphs is known to exclude \cite{cera2007graphs}: Wagner's theorem established that planar graphs cannot contain \Kn{5} or \bip{3}{3} minors \cite{Wagner1937}, while the works of Robertson and Seymour \cite{ROBERTSON198339,ROBERTSON198692,ROBERTSON1990255,Robertson200343,Robertson2004325} develop the theory of forbidden minors for planar and non-planar graphs (see also the review in Ref.~\cite{Kawarabayashi2007}). Extremal graph theory identifies what classes of graphs, or graph qualities, ensure certain minors are contained by a graph \cite{song2006extremal,Fountoulakis2007,Fountoulakis2009,joret2013complete,eppstein2013grid}. We will show that the final minor in the MSC of a complete bipartite graph is always the \Kn{N+1} graph, and in establishing the robustness of this minor for more general bipartite graphs, we turn to research on the development of theorems for the existence of complete minors: \cite{Kostochka1984,thomason2001extremal,diestel2004dense,thomason2008disjoint,bohme2009linear}. However the class of bipartite graphs under consideration are not particularly sparse, nor are they random, of large order, size, girth or degree. In focusing on the \Kn{N+1} minor, we are searching for conditions that ensures a graph has the largest possible complete minor. 

Currently our method for constructing the MSC of a graph is only applied to the case of complete or near-complete bipartite graphs and in this paper we focus on covering non-planar bipartite graphs which are undirected and contain no multiple edges or self-loops. This excludes several cases of bipartite graphs which have a trivial MSC that only contains the original graph. For example, simple paths, cycles, and the star graph $\g{S}_n$. We also identify leaves (terminal vertices) as graph edges which do not contribute to the formation of a member of the MSC. 

We derive the necessary requirements to build an edge set which generates a MSC beginning with the simplest case of the complete (fully connected) bipartite graph (Sec.~\ref{sec: c-c complete bipartite}). A complete bipartite graph \bip{N}{N} which has minimum partition order $N$ has a MSC which contains the complete graph \Kn{N+1}. This is the largest complete graph that can be embedded into the original graph as no other set of edge contractions can lead to a minor with a completely disconnected graph complement. For the \Ch{n}{n}{c} hardware graph,the MSC identifies \Kn{nc+1} as the largest complete graph which can be minor embedded into the virtual hardware \bip{nc}{nc}. This agrees with results found by Choi \cite{choi2011minor} and results given in  \cite{klymko2014maximal} based on treewidth arguments. 

Incomplete bipartite graphs are those graphs which are missing edges between partitions. We discuss the case of a complete bipartite graph missing a small number of edges in Sec.~\ref{sec:incomplete bipartite graphs}, and focus on the \Kn{N+1} minor robustness on a general incomplete bipartite graph. Criteria are derived which identify bipartite graphs of minimum partition order $N$ which lack the minor \Kn{N+1}.

\section{\label{sec:defs}Definitions and notations}

A graph $\g{G} = \g{G}(V,E)$ is defined by a vertex set $\lbrace v_i \rbrace \equiv V$ and an edge set $\lbrace e_{ij} \rbrace \equiv E$. 
In this paper we only consider simple, undirected graphs: the edges have no orientation (the edge $e_{ij} = (x_i, x_j) $ is equivalent to the edge $e_{ji} = (x_j,x_i)$), and multiple edges and self-loops are not allowed. When counting the degree of a vertex set resulting from the contraction of an edge $e_{ij}$, the \textit{in-degree} of vertex $x_i$ counts all edges which connect to $x_i$ but excludes the edge $e_{ij}$, and the \textit{out-degree} of vertex $x_j$ counts all edges which connect to $x_j$ but excludes the edge $e_{ij}$. 

For any graph \g{G} on $n$ vertices, the complement graph $\g{G}^{c}$ is also defined on the vertex set $V(\g{G})$ and contains an edge $\overline{e}_{ij}$ only if $\overline{e}_{ij}$ does not exist on \g{G}. Thus the union of \g{G} and its complement $\g{G}^{c}$ form the complete graph on $V(\g{G})$. For example, the complete graph $\g{G} = \Kn{n}$ has a complement of order $|V(\g{G}^{c})|= |V(\g{G})|=n$ vertices but $|E(\g{G}^{c})|= 0$ edges. 

A graph $\g{G^{\prime}}$ is \textit{minor embeddable} in \g{G} if for each vertex $v$ of a graph $\g{G^{\prime}}$ a mapping $\phi_{\g{G}}(v)$ can be found which takes $v$ to a connected subtree of \g{G}. The individual vertex sets of $\phi_{\g{G}}(v)$ do not overlap, and $\phi_{\g{G}}(x), \phi_{\g{G}}(y)$ are adjacent if there exists vertices $x_i \in \phi_{\g{G}}(x)$ and $y_i \in \phi_{\g{G}}(y)$ which are adjacent on \g{G}. The order of a given subtree is the number of vertices it contains. An isomorphic embedding maps each vertex $v$ to a single vertex of \g{G} (i.e. $\phi_{\g{G}}(v) = v$), i.e. each vertex is mapped to a subtree of order $1$. 

Generating a minor (\g{M}) from a given graph \g{G} is done by edge contraction or edge removal, and \g{G} may have a large number of minors. For finite graphs, there exists a set of graph minors $\g{M} \equiv \lbrace \g{M}^{(i)} \rbrace $ which we define as the MSC. The minors $\g{M}^{(i)} \in \g{M}$ are unique in that they are not isomorphically embeddable in the original graph, nor are they contained in any other minor as a subgraph. As a result, the set \g{M} covers the entire set of minors (i.e for any minor $m$ of \g{G}, $m$ is either a member of the MSC or is contained in a minor of the set cover as a subgraph). By definition any minor formed by edge deletion cannot be a included in the MSC. When searching for possible members of the MSC only minors formed by edge contraction are considered.

This work focuses on the MSC for bipartite graphs (\bip{N}{N^{\prime}}). A bipartite graph is a graph with a vertex set which can be partitioned into two non-overlapping subsets: $v_{a} \subset V(\g{G}), v_{b} \subset V(\g{G}), v_{a} \cup v_{b} = V(\g{G}), v_{a} \cap v_{b} = \varnothing$. Complete bipartite graphs \bip{N}{N^{\prime}} are those graphs of size $|E(\bip{N}{N^{\prime}})| = N N^{\prime}$, with all edges existing between vertices in different partitions. Incomplete bipartite graphs \incbip{N}{N^{\prime}} are graphs with missing edges between partitions, $|E(\incbip{N}{N^{\prime}})| < N N^{\prime}$. On a complete bipartite graph, all vertices in a partition have the same degree, which is equal to the order of the other partition. For the case of incomplete bipartite graphs, we define vertices which are not fully connected to the opposite partition as \textit{incomplete vertices}. 

A subset of the edge set of a graph, is called an \textit{edge matching} if all edges are non-adjacent (do not share a vertex). A perfect matching is an edge set which leaves no vertex of the graph uncoupled. The size of the perfect matching set for a complete bipartite graph \bip{N}{N} is $N$ and is an upper bound for the size of the perfect matching on an incomplete bipartite graph \incbip{N}{N}. For a complete bipartite graph with unequal partition orders, \bip{N}{N^{\prime}}, the size of the perfect matching set is $\min(N,N^{\prime})$. 

\section{\label{sec:embedding in bipartite virtual hardware} Minor embedding in the ideal Chimera hardware graph}

The quantum annealer from D-Wave Systems, Inc.\ uses a lattice of coupled superconducting flux qubits. The topology for the connections and interactions between the qubits is represented by a hardware graph referred to as the \Ch{n}{m}{c} graph \cite{url_DWAVE,Johnson2011,Dickson2013,PhysRevX.4.021008}. This graph has a fixed topology: it is an $n \times m$ square lattice of \bip{c}{c} unit cells with intercell connections column-wise between left bipartite partitions and row-wise between right bipartite partitions (see \Ch{3}{3}{4} in Fig.~\ref{fig:chimera}). The \Ch{n}{m}{c} is bipartite, and the largest biclique is the unit cell \bip{c}{c}. Few problems can be isomorphically embedded into the D-Wave processor. In general, optimizing a given logical Hamiltonian requires the use of minor embedding into the hardware graph, which creates a significant bottleneck in the quantum annealing workflow. 

Rather than enumerate the entire MSC of the Chimera graph, we use an intermediate embedding step to construct a virtual hardware graph of the original hardware graph. We assume the Chimera graph is ideal, with all qubits operational (no hard faults). Contracting all intercell connections on a \Ch{n}{m}{c} hardware results in a virtual hardware\ref{note1} which is a complete bipartite graph \bip{nc}{mc} (see Fig.~\ref{fig:chimera}). For this virtual hardware, we then construct the MSC. The remainder of this paper studies the construction of the MSC and establishes that the largest complete graph embeddable in a \bip{nc}{mc} virtual hardware is \Kn{\min{(n,m)}c+1}.  These results agree with lemmas presented in Ref.~\cite{klymko2014maximal} which showed the treewidths of \bip{c}{c} and \Kn{c} coincide, $\tau\left(\bip{c}{c}\right) = \tau\left(\Kn{c+1}\right) = c$. 
\begin{figure}[htbp]
  \includegraphics[width=0.9\columnwidth]{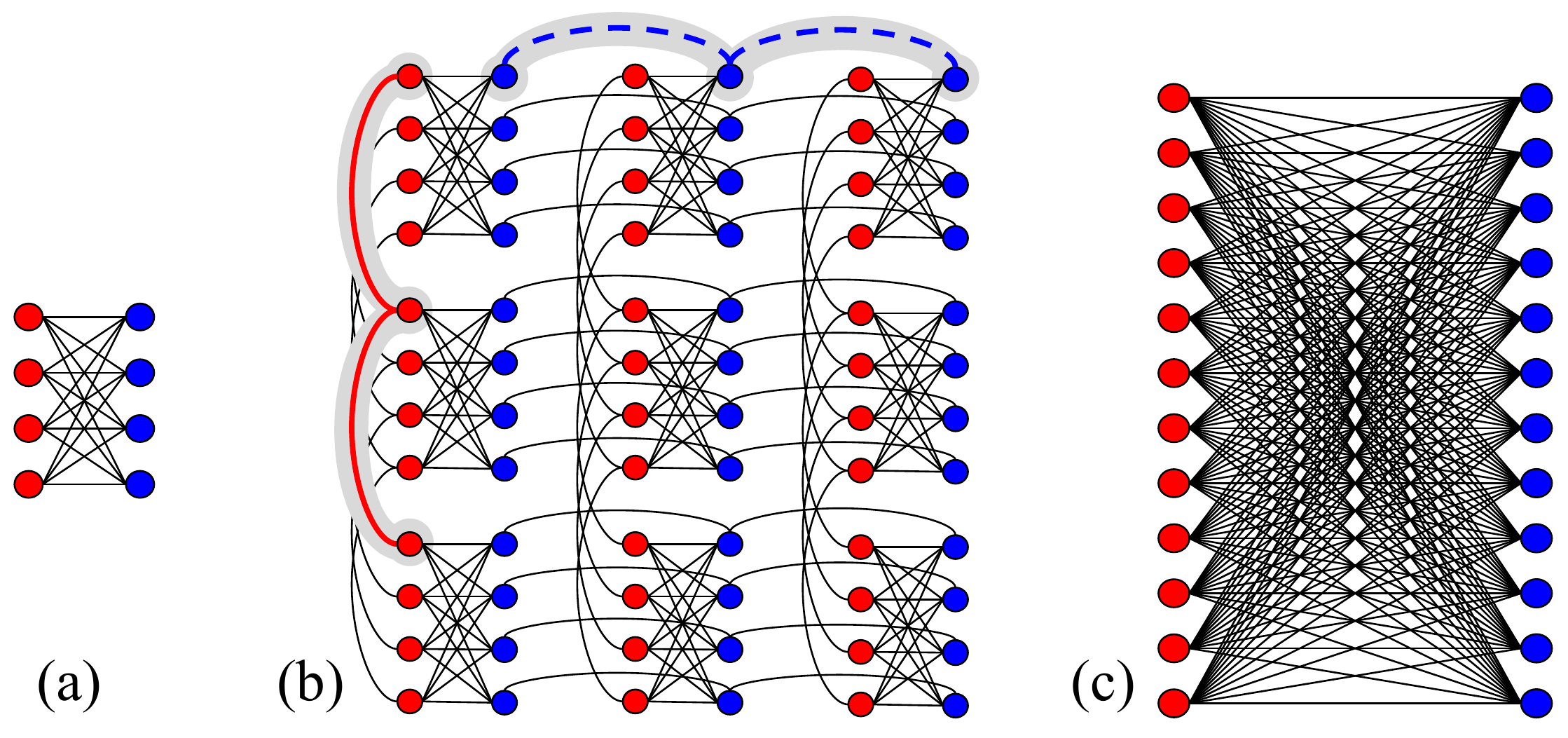}
  \caption{[Color online] The Chimera graph and virtual hardware construction: (a) the unit cell of a Chimera graph is the \bip{4}{4} graph, (b) the $3 \times 3$ grid of unit cells for \Ch{3}{3}{4}, vertex bags are defined along each row or column of intercell connections (highlighted in grey), (c) contraction along all edges in a bag results in the \bip{12}{12} virtual hardware. The left partition consists of vertex bags with $3$ physical qubits, formed by contracting vertical intercell connections (red), the right partition consists of vertex bags with $3$ physical qubits, formed by contracting all horizontal intercell connections (blue, dashed)}
  \label{fig:chimera}
\end{figure}
The approach of Klymko et al.\ was to use the maximal minor set of the Chimera unit cell to define an iterative embedding process. In this work we consider a two-step embedding procedure which first generates a complete bipartite virtual hardware then identifies the MSC of the virtual hardware. An embeddable complete graph \Kn{n} is embedded into the virtual hardware through an isomorphic mapping into (at least) one of the minors in the minor set.  

For an ideal \Ch{n}{n}{c} the two-step embedding and the Klymko embedding identify the largest embeddable complete graph as \Kn{nc+1}. Generalizing to a \Ch{n}{m}{c} hardware, the largest embeddable graph is \Kn{\min{(n,m)c}+1}. Using the two-step embedding results in a final embedding for \Kn{\min{(n,m)c}+1} which embeds each logical qubit into a chain of length $m$, $n$ or length $m+n$. With the MSC construction, the embedding of \Kn{nc+1} into a Chimera hardware graph with $n\times n$ square grid of \bip{c}{c} unit cells will require $2$ chains of order $n$ and $cn-1$ chains of order $2n$.  The distribution of chain lengths for the embedding of \Kn{13}  on a $3 \times 3 $ square Chimera is: $2$ chains of order $\ell = m = 3$ and $11$ chains of order $\ell=2m = 6$. Comparison to the embeddings shown in Ref. \cite{klymko2014maximal}, where the maximal minor set of the Chimera unit cell was used to define an iterative embedding process, shows that the two-step embedding finds the same embeddings for \Kn{13} and \Kn{17} as the iterative minor extension embeddings. 
\begin{figure}[htbp]
  \includegraphics[width=1.00\columnwidth]{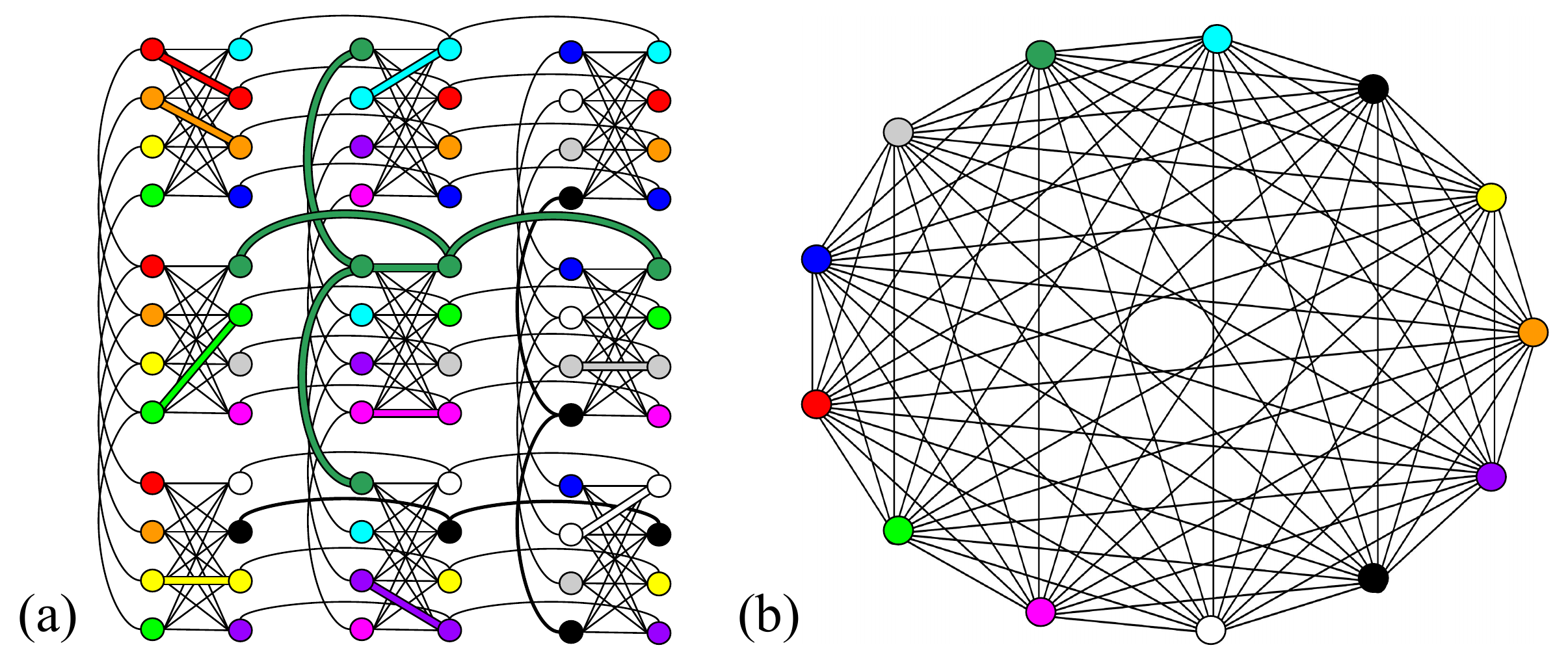}
  \caption{[Color online] Example of embedding \Kn{13} in the \Ch{3}{3}{4} virtual hardware construction:  (a) the embedding of a logical qubit into $6$ physical qubits is highlighted (dark green), (b) of the 13 vertices on \Kn{13}, $11$ are chains of order $6$ (various colors) while $2$ are chains of order $3$ (black)}
  \label{fig:K13_embedding}
\end{figure}

In Ref.~\cite{choi2011minor}, Choi gave a lower bound of $\lceil \frac{n-3}{d-2} \rceil$ on the minimum number of qubits needed to embed each logical qubit, and a lower bound on the total number of qubits needed to embed a graph of order $n$: $\Omega (n^2/d)$ qubits, where $d=c+2$ is the number of couplings per physical qubit. Our embedding on \Ch{n}{n}{4} for \Kn{4n+1}, by embedding each logical qubit into ether $n$ or $2n$ physical qubits, saturates the lower bound $\lceil \frac{nc-3}{d-2} \rceil = \lceil \frac{4n - 3}{4}\rceil = n$. The ideal \Ch{12}{12}{4} hardware contains $1152$ qubits, while the minimum number of qubits needed to embed \Kn{49} is 400 qubits. In this work we show that by embedding into the \bip{48}{48} virtual hardware, the largest embeddable clique is \Kn{49}.

The remainder of this section is centered around proving the following theorem:
\begin{theorem}
On the ideal \Ch{n}{m}{c} hardware ($n \times m$ grid of \bip{c}{c} unit cells) the largest complete graph which can be embedded through the two-step process is \Kn{d_{min}c+1}, where $d_{min} = \min{(n,m)}$.
\end{theorem}
\begin{proof}
Through the construction of a bipartite virtual hardware, each physical qubit is contained in a virtual qubit, which is a chain of length $n$. It will be shown in the construction of the MSC that each logical qubit eventually is embedded into a single virtual qubit (final embedding into a subtree of order $n$) or at most a pair of virtual qubits (final embedding into a subtree of order $2n$). 
\end{proof}
\begin{lemma}
On the ideal \Ch{n}{n}{c} hardware, embedding the complete graph \Kn{nc+1} will have $nc-1$ subtrees of order $2n$ and $2$ subtrees of order $n$.
\end{lemma}
The dimensions of the quantum hardware can be used to minimize the order of the vertex subtrees. A graph \Kn{N+1} can be embedded in \Ch{n^{\prime}}{n^{\prime}}{c^{\prime}} hardware if $n^{\prime}c^{\prime} \geq N$ but the order of each subtree is only dependent on $n^{\prime}$. For example, \Kn{41} can be embedded in a \Ch{10}{10}{4} hardware with maximum embedding subtree of order $20$, or can be embedded in \Ch{2}{2}{20} hardware with maximum embedding subtree of order of $4$. 

\section{\label{sec: c-c complete bipartite} MSC of complete bipartite graphs}
The MSC of a complete bipartite graph \bip{N}{N} contains exactly $N$ minors, and creates a graph sequence which converges to a complete graph $\Kn{N}^{\prime}$, ($N^{\prime}<2N$). The edges which one must contract in order to form minors of the MSC are a set of $(N-1)$ non-adjacent edges. We prove our construction of the MSC, and the identification of the set of edges to contract in terms of an edge matching set and the graph complement of the complete bipartite graph. 

Based on our definition of the MSC as the minors which covers all possible minors which can be constructed from a graph, we expect each set member to be the densest connected graph on $N^{\prime}$ vertices. As each MSC minor is denser than the previous minor, the corresponding sequence of graph complements will become sparsely connected. The final minor in the set cover of a complete bipartite graph \bip{c}{c} is found to be the complete graph \Kn{c+1} whose complement graph is of order $(c+1)$ but size $0$.
\begin{figure}[htbp]
  \includegraphics[width=0.9\columnwidth]{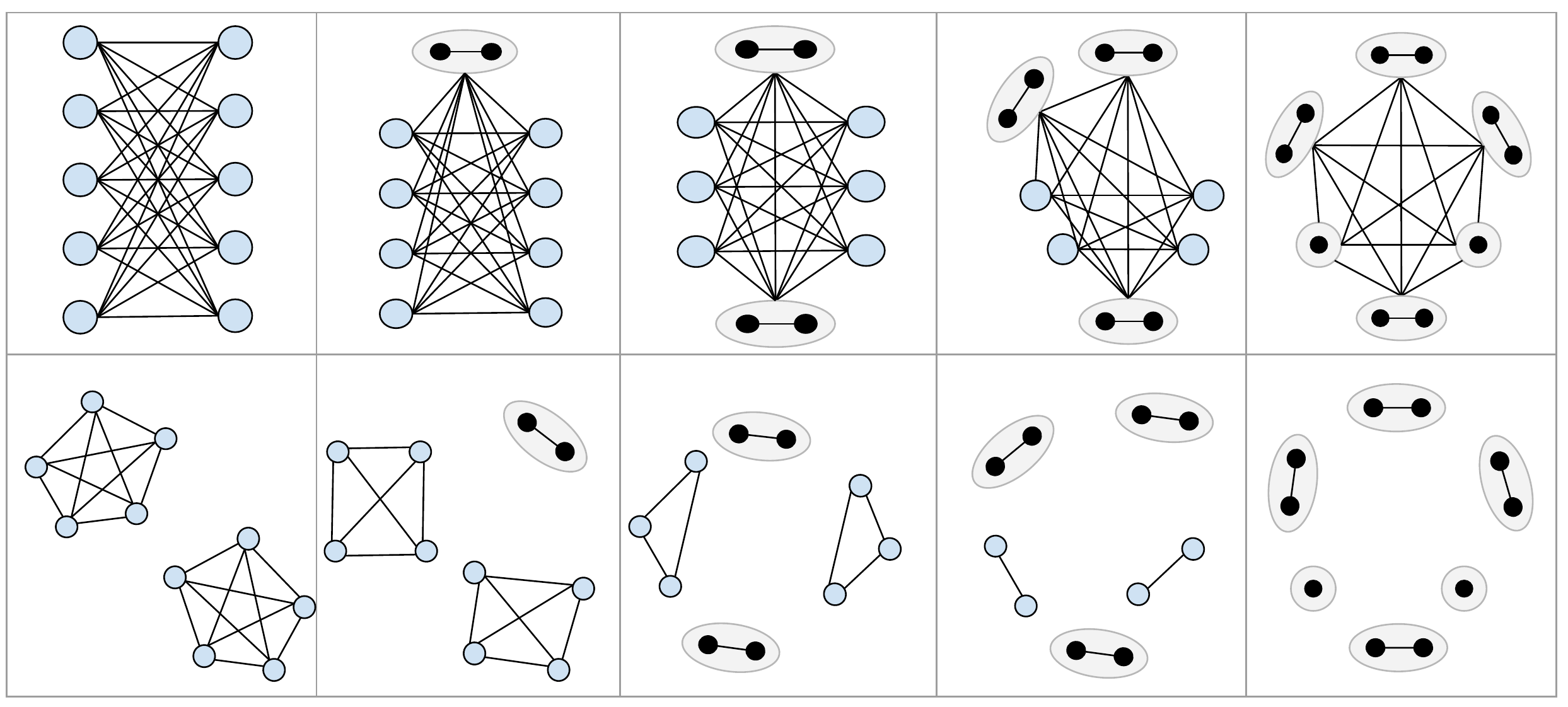}
  \caption{Top row: the construction of the MSC for a complete bipartite graph \bip{5}{5}. Bottom row: the evolution of the complement for each minor. The final MSC minor is the complete graph \Kn{6} which has a corresponding complement graph which has an empty edge set}
  \label{fig:complete_bipartite_complement_series}
\end{figure}
\begin{theorem}
\label{th:edge-contraction}
For a complete bipartite graph $\g{G} =\bip{c}{c}$, contracting a single edge will result in a minor $\g{M}^{(1)}_{\g{G}}$ which is a member of the MSC.
\end{theorem}
\begin{proof}
By definition, a bipartite graph does not have connections between vertices contained in the same vertex partition. Contracting any edge on a complete bipartite graph will result in a vertex bag which is connected to both bipartite partitions. Consider an edge on the bipartite graph $e_{ij}$ connecting vertex $v_i$ in the left partition and vertex $v_j$ in the right partition. We also define the sets of $(N-1)$ neighboring vertices $\lbrace v^{\prime}_i \rbrace, \lbrace v^{\prime}_j \rbrace$. All vertices in $\lbrace v^{\prime}_i \rbrace$ are in the right partition, and vice versa for $\lbrace v^{\prime}_j \rbrace$. The contraction of $e_{ij}$ creates the vertex bag $\phi (x_{ij})$ which is fully connected: having connections to all vertices in $\lbrace v^{\prime}_i \rbrace \cup \lbrace v^{\prime}_j \rbrace$. Consequently, such a minor cannot exist as a subgraph of the original graph because of its degree $2(N-1)$ and its connections between both partitions.
\end{proof}
It is also seen that the contraction of an edge on the original bipartite graph and the creation of a minor with a fully connected vertex results in a disconnected vertex in the minor's complement. Next we consider the action of contracting two edges on a graph. First, considering the case of edges which are non-adjacent on the original graph (do not share any vertices).
\begin{corollary}
\label{cor: non-adjacent-edge contraction} 
The contraction of a pair of non-adjacent edges $(e_1, e_2)$ on a complete bipartite graph \bip{c}{c} will result in a MSC minor. 
\end{corollary}
\begin{proof}
The contraction of the first edge $e_1$ will result in a MSC minor $\g{M}^{(1)}_{\g{G}}$ by creating a fully connected vertex. Contracting the second edge $e_2$ will also create a fully connected vertex. As a result, the minor created by contracting two non-adjacent edges $\g{M}^{(2)}_{\g{G}}$ is identified as a member of the MSC: it cannot be contained in  $\g{M}^{(1)}_{\g{G}}$ as a subgraph due to the additional fully connected vertex, nor is it possible for it to be a subgraph of $\g{M}^{(0)}$.
\end{proof}
When non-adjacent edges are contracted we see that another disconnected vertex exists on the minor's complement. This effect is not observed for the minor formed when two adjacent edges are chosen on the original graph and then contracted.
\begin{corollary}
\label{cor:adjacent-edge contraction} 
The contraction of a pair of adjacent edges on a complete bipartite graph will not result in a MSC minor. Contracting adjacent edges will result in a graph minor which is a subgraph of the MSC minor generated by contracting non-adjacent edges. Likewise the complement of such a minor will contain the MSC minor complement as a subgraph. 
\end{corollary}
\begin{proof}
Consider two sets edges on a complete bipartite graph. Set 1 contains a pair of non-adjacent edges ($S_1 = \lbrace (v^{a}_{1}\leftrightarrow v^{b}_{1}),(v^{a}_{2} \leftrightarrow v^{b}_{2})\rbrace$); the edges connect four distinct vertices on the graph with 2 vertices in each bipartite partition. Set 2 contains a pair of adjacent edges; they connect 3 distinct vertices on the graph ($S_2 = \lbrace (v^{a}_{1} \leftrightarrow v^{b}_{1}),  (v^{a}_{2} \leftrightarrow v^{b}_{1}) \rbrace$). The first contraction in either set will add connections between vertices in the same partition: $v^{a}_{1} \leftrightarrow v^{a}_{2}, v^{a}_{3},\dots v^{a}_{c}$, $v^{b}_{1} \leftrightarrow v^{b}_{2}, v^{b}_{3},\dots v^{b}_{c}$. Contracting a non-adjacent edge adds the (unique) connections: $v^{a}_{2} \leftrightarrow  v^{a}_{3},\dots v^{a}_{c}$, $v^{b}_{2} \leftrightarrow v^{b}_{3},\dots v^{b}_{c}$, but contracting an adjacent edge only adds the connections: $v^{a}_{2} \leftrightarrow v^{a}_{3},\dots v^{a}_{c}$.
\end{proof}

The results for non-adjacent edges are extended further, establishing that there is a finite size edge matching which will generate all minors of the MSC for a complete bipartite graph.
\begin{theorem}
\label{th:complete set_maxminors_Kcc}
For a complete bipartite graph \bip{N}{N} the cardinality of the MSC is equal to $n$ and is formed by the contraction of a set of $(n-1)$ non-adjacent edges where $n$ is the size of the maximal edge matching set.
\end{theorem}
\begin{proof}
By our definition, a MSC minor of a graph cannot be contained within any other minor as a subgraph. Any graph has at least one minor in its MSC, $\g{M}^{(0)} = \g{G}$ the original graph. All remaining minors in the MSC must be non-isomorphic to each other, any other minor, or the original graph. 

From Theorem (\ref{th:edge-contraction}), we established that contraction of a single edge on a bipartite graph creates the first minor $\g{M}^{(1)}$ of order $2N-1$, with $1$ vertex set of degree $2N-2$ and the remaining vertices with degree $N$. It was shown in Corollary (\ref{cor: non-adjacent-edge contraction}) that a pair of non-adjacent edge contractions will create a second minor $\g{M}^{(2)}$ of order $2N-2$ with $2$ vertex set of degree $2N-3$ and the remaining vertices with degree $N$. 

This argument is extended to postulate that $N$ non-adjacent edges could create $N$ minors in the MSC where the value of $N$ is determined by enforcing the condition that all minors in the MSC are non-isomorphic. As edges are contracted, the order of each subsequent minor is reduced by $1$, while the number of fully connected vertices is increased by $1$. After contracting $(N-1)$ edges the minor $\g{M}^{(N-1)}$ is of order $N+1$ and has $(N-1)$ vertex sets of degree $2N-N = N$ and $2$ vertices of degree $N$ (the complete graph \Kn{N+1}). While additional non-adjacent edges may exist, any further contractions will result in minors which can be contained in the $\g{M}^{(N-1)}$ minor. For the complete bipartite graph, $n=N$ the size of the maximal edge matching set.

From choosing $(N-1)$ edges to contract and including the first minor defined by the original graph, it follows that any complete bipartite graph \bip{N}{N} will have a MSC of cardinality $N$. 
\end{proof}
\begin{corollary}
\label{cor: K11}
The complete bipartite graph \bip{1}{1} has MSC of cardinality $1$. 
\end{corollary}
\begin{corollary}
\label{cor: star graphs}
A complete bipartite graph which is also a star graph $\bip{N}{1} = \mathcal{S}_N$ has MSC of cardinality $1$.
\end{corollary}
From the above theorems we present our final result for complete graphs.
\begin{theorem}
\label{cor:maximum clique minor}
The largest clique minor of a \bip{N}{N} graph is \Kn{N+1}.
\end{theorem}
\begin{proof}
The original bipartite graph \g{G} =\bip{N}{N} has $2N$ total vertices of degree $d = N$. The MSC of such a graph is the ordered sequence formed by contracting $(N-1)$ non-adjacent edges: $\lbrace \g{M}^{(0)}, \g{M}^{(1)}, \g{M}^{(2)},\dots,\g{M}^{(N-1)} \rbrace $, where $\g{M}^{(0)} = \g{G}$. The order of minor $\g{M}^{(i)}$ is determined by the previous minor: 
\begin{equation}
|V(\g{M}^{(i)})| = |V(\g{M}^{(i-1)})| -1 =  |V(\g{M}^{(0)})| - (i).
\end{equation}
The number of vertices of degree $N$ is reduced by 2 on each minor while the number of fully connected vertices is increased by 1. As a result the final MSC minor $\g{M}^{(N-1)}$ is of order $2N - (N-1) = N+1$ and $N$-degree regular, it is the graph \Kn{N+1}. By definition the graph \Kn{N+2} cannot exist or it would be in the MSC and the graph \Kn{N} is contained within \Kn{N+1} as a subgraph.
\end{proof}
\begin{corollary}
The largest clique minor of a \bip{N}{N^{\prime}} graph is \Kn{\min{(N,N^{\prime})}+1}.
\end{corollary}
\begin{corollary}
The clique number of each minor is strictly increasing over the MSC of a complete bipartite graphs.
\end{corollary}
\begin{proof}
The first minor in the MSC is the bipartite graph, which has a clique number of $2$. As non-adjacent edges are contracted, each subsequent minor has an increasing number of vertices which are connected to both bipartite partitions. As a result the order of the largest clique in each minor increases by $1$ vertex and the clique number increases by $1$. The maximum clique number of \bip{N}{N} is $N+1$ and is reached on the final minor in the MSC.
\end{proof}
\begin{theorem}
\label{th: bounded treewidth}
The treewidth of any minor in the MSC of a complete bipartite graph \bip{N}{N} is N.
\end{theorem}
\begin{proof}
The proof follows from a simple treewidth argument. Each minor of the \bip{N}{N} MSC is formed sequentially, so each minor in the MSC is a minor of the previous member. By definition, a graph \g{H} is a minor of \g{G} if the treewidth of \g{H} is bounded from above by the treewidth of \g{G}: $\mathbf{tw}(\g{H}) \leq \mathbf{tw}(\g{G})$. The treewidth of the first minor, a complete bipartite graph \bip{N}{N}, is known $\mathbf{tw}(\bip{N}{N}) = N$. The treewidth of the final minor \Kn{N+1} is known $\mathbf{tw}(\Kn{N+1}) = N$ and thus each minor in the remaining sequence must have treewidth $N = \mathbf{tw}(\bip{N}{N}) \geq \mathbf{tw}(\g{M}^{(i)}) \geq \mathbf{tw}(\Kn{N+1})  = N$.
\end{proof}

The construction of the MSC for a complete bipartite graph \bip{N}{N} can be implemented using a simple greedy algorithm. To add an edge to the matching set, one is chosen at random from the existing edges of the graph. The first edge is chosen from all edges of the graph. To ensure subsequently chosen edges are not adjacent to any edge which already exists in the set $E_{\g{M}}$, after adding an  edge to the minor cover edge set, all edges adjacent the head and tail of an edge $e_{ij}$ are removed from the graph. This process is repeated, with edges added to the matching set until a stopping condition is met. Once all edges in the non-adjacent edge set are contracted, and vertex sets formed, the remaining vertices of the original graph are isomorphically mapped to vertex sets of size 1. 

As stated in Corollary ~\ref{cor: K11}, the \bip{1}{1} graph minor has $|\g{M}|=1$. This is used to define the stopping condition: once the removal of the head/tail vertices of a contracted edge results in a subgraph which is \bip{1}{1} the procedure ends. It is seen that this procedure constructs a set of only $(N-1)$ non-adjacent edges, from which all possible members of the MSC are formed.   

\section{\label{sec:incomplete bipartite graphs} MSC of nearly complete bipartite graphs}
Most bipartite graphs encountered in real world applications are not complete, it is likely that connections are absent between the partitions. For example, a quantum processor may have faulty (inoperable) qubits, and thus the \Ch{n}{m}{c} hardware graph may be missing vertices.  In this section we look at the MSC for the class of incomplete bipartite graphs \incbip{N}{N}, created by removing a sparse subset of edges from a complete bipartite graph. The approach outlined in Sec.~\ref{sec: c-c complete bipartite} may have limited applicability for such graphs. Randomly choosing connections on a graph may not result in a set of $(N-1)$ non-adjacent edges; collisions between the missing edge sets and the non-adjacent edge set are probable and the completion of the non-adjacent edge set may require the addition of an edge which does not exist on the graph. Additionally, even if a set of $(N-1)$ non-adjacent edges is found, there is no guarantee the contractions will yield the clique minor \Kn{N+1}. Since random bipartite graphs, and incomplete bipartite graphs are very common in quantum annealing (discussed in Sec.~\ref{sec:embedding in bipartite virtual hardware}) this last point is very important to investigate.

The MSC of \bip{N}{N} establishes that the largest embeddable clique is \Kn{N+1}. On an incomplete bipartite graph, the edges which contract and form members of the MSC cannot be chosen completely at random. For \incbip{5}{5} missing a single edge, we show in Fig.~\ref{fig:incomplete_bipartite_maximal_minors} it still contains the clique minor \Kn{6} if the non-adjacent edge set is chosen appropriately. However this result is not guaranteed for any arbitrary non-adjacent edge set (as shown in Fig.~\ref{fig:incomplete_bipartite_minors1}). The robustness of the \Kn{N+1} minor as edges are removed from \bip{N}{N} is of importance in quantum annealing applications (discussed in Sec.~\ref{sec:embedding in bipartite virtual hardware}) where hard faults can dramatically affect the connectivity of the hardware graph. 

We show how the results in Sec.~\ref{sec: c-c complete bipartite} can be modified and applied to the simple case of \bip{N}{N} missing a single edge. From there we define three conditions which identify graphs \incbip{N}{N} that cannot have a clique minor of order $N+1$. Only graphs with equal partition orders are discussed, but the results can be extended to \incbip{N}{N^{\prime}}. However, a full discussion of how to mitigate the effects of faults is left as an open question.
\begin{figure}[htbp]
  \includegraphics[width=0.9\columnwidth]{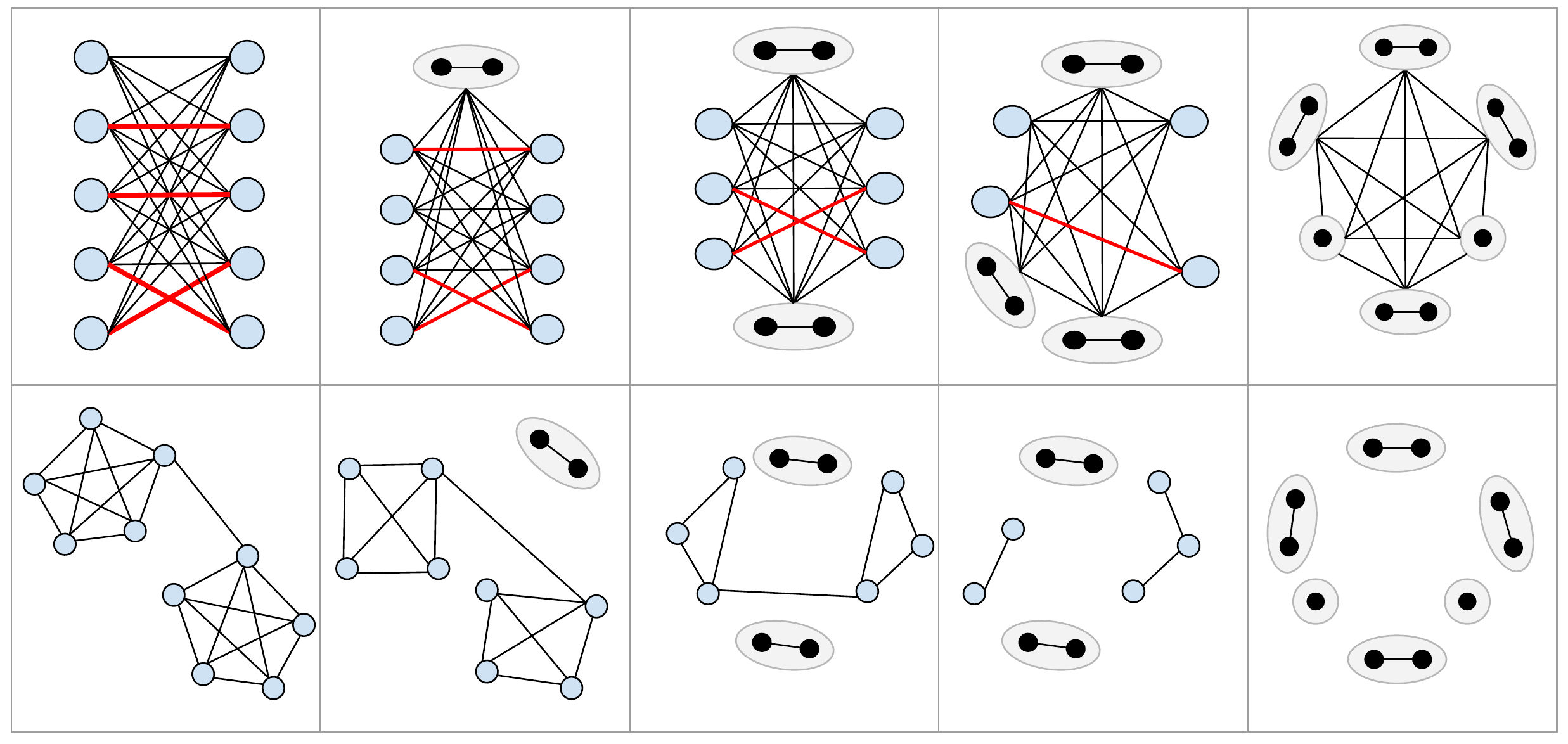}
  \caption{[Color online] Top row: evolution of a minor set generated by a non-adjacent edge set for \incbip{5}{5} (missing a single edge) which covers the set of incomplete vertices (heavy red lines). Bottom row: evolution of the complement for each minor}
  \label{fig:incomplete_bipartite_maximal_minors}
\end{figure} 

Most notably there exist a class of edges which will never assist in the creation of a MSC minor. These are identified as those edges connecting to leaves (also known as terminal vertices).
\begin{theorem}
\label{theor:leaves contraction} 
Contraction of an edge which connects to a terminal vertex will never yield a minor in the MSC.
\end{theorem}
\begin{proof}
In our definition of the MSC, we noted that any member of the MSC cannot be contained in another member as a subgraph. Consider a graph \g{G} with a leaf vertex. The edge connecting to the leaf has a head with in-degree ($d_i -1$) while the tail has out-degree $0$. Contraction of this edge will result in a vertex set $\phi(x)$ with degree $(d_i -1)$ and the resulting minor will be a subgraph of \g{G}.
\end{proof}
\begin{figure}[htbp]
  \includegraphics[width=0.9\columnwidth]{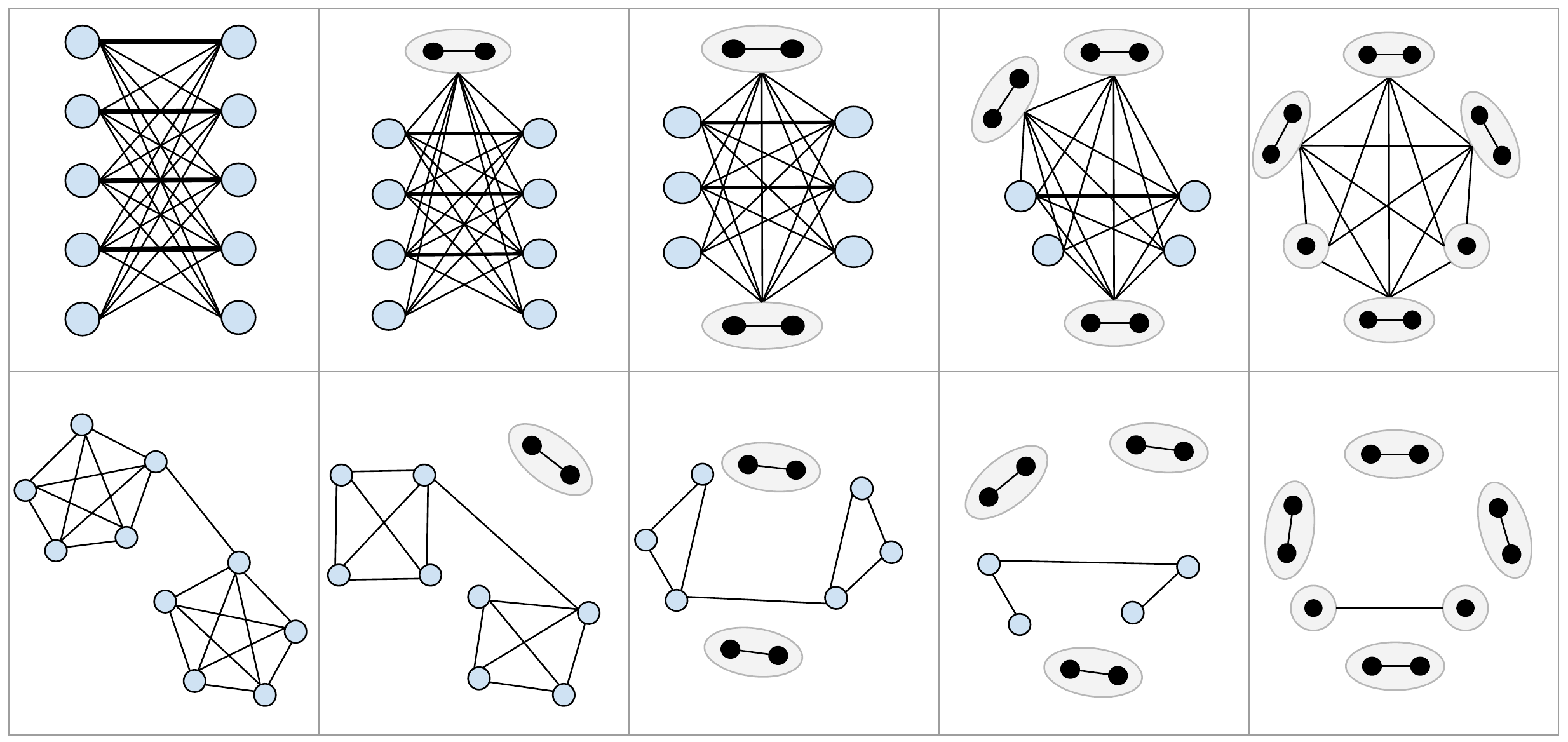}
  \caption{Top row: evolution of a minor set generated by an arbitrary non-adjacent edge set (heavy black lines) for \incbip{5}{5} (missing a single edge) . Bottom row: evolution of the complement for each minor}
  \label{fig:incomplete_bipartite_minors1}
\end{figure} 
However, the \Kn{N+1} clique minor can only be generated after $(N-1)$ edge contractions for either a complete bipartite graph \bip{N}{N} or incomplete bipartite graph \incbip{N}{N}. The robustness of the \Kn{nc+1} graph minor is given by the following result:
\begin{theorem}
\label{th: large clique minor robustness}
On the faulty \Ch{n}{n}{c} hardware, embedding the complete graph \Kn{nc+1} is not possible if the virtual hardware is missing more than $m^{\prime} = (nc)(nc+1)/2 + (nc-1)$ edges and the set of $n^{\prime}$ incomplete vertices cannot be covered by $(nc-1)$ edges.
\end{theorem}

Our proof of the clique minor robustness is framed in terms of the complement graph. For a complete bipartite graph \bip{N}{N}, its complement graph will consist of two disconnected complete graphs \Kn{N}. As $\bip{N}{N} \rightarrow \incbip{N}{N}$ through the removal of edges while maintaining partition orders, the complement graph will gain edges which connect the two disconnected \Kn{N} graphs.

The first condition for a faulty Chimera graph is that it's bipartite virtual hardware cannot be missing more than $(N+1)N/2 + (c-1)$ edges. This is a resource argument, that the original graph must have a sufficient number of edges to contract and still support the complete graph \Kn{N+1}. The second requirement is that the set of incomplete vertices on the bipartite graph must be coverable by at most $(N-1)$ edges. If an incomplete vertex is not contained in a larger vertex embedding, then the resulting minor will not have the maximum degree possible. The third requirement is that throughout the sequence of edge contractions, a disconnected dimer consisting of two vertex embeddings joined by an edge must not be created in the minor complement. If such a dimer exists it cannot be resolved into two disconnected vertex sets. 
\begin{example}
\label{ex: inc bip 22}
Consider the bipartite graph formed by removing a single edge from \bip{2}{2}. The clique minor \Kn{3} cannot be formed from the resulting incomplete bipartite graph \incbip{2}{2}  because the original graph has $3$ edges, any contraction would result in a minor with $3$ vertices and $2$ edges.
\end{example}
\begin{example}
\label{ex: crown graph}
The crown graph is an incomplete bipartite graph which is formed from \bip{N}{N} by removing the $N$ edges of a perfect matching and is an example of a graph which does not contain a \Kn{N+1} minor. There are $2N$ incomplete vertices, this set cannot be covered by any combination of $(N-1)$ edges. The lack of a \Kn{N+1} minor is further verified by treewidth argument (see \cite{bodlaender1994tourist},\cite{Fomin2003}): a graph \g{G} contains \g{H} as a minor if the treewidth of \g{H} is bounded above by the treewidth of \g{G}, $\mathbf{tw}(\g{H}) \leq \mathbf{tw}(\g{G})$. The crown graph has treewidth $N-1$ and thus cannot embed \Kn{N+1} ($\mathbf{tw}(\Kn{N+1}) = N$).
\end{example}

\section{Open problems}
We close with a brief discussion of two open problems: the more full treatment of a faulty hardware graph and the full implementation of MSC embedding.  

Hard faults (inoperable qubits) can occur in a quantum annealer, which result in missing vertices and edges on the hardware graph. There are many approaches to mitigating the effects of hard faults: one can choose to consider only the largest $n^{\prime} \times n^{\prime}$ square portion of the hardware graph which lacks any faults (see \cite{klymko2014maximal}) or distort the shape of vertex bags to accommodate faults (see \cite{zaribafiyan2016systematic}).  In our preliminary discussion, we consider simply removing any hard faults prior to constructing a virtual hardware. The resulting virtual hardware of a \Ch{n}{n}{c} with faults will remain a bipartite graph, however it will be missing edges, the partition orders may not be equal, nor equal to $nc$, and its associated MSC may be larger than that of \bip{nc}{nc}. The robustness of the \Kn{N+1} clique minor was discussed in Sec.~\ref{sec:incomplete bipartite graphs} for graphs missing few edges and the construction of edge matchings needs careful consideration to maximize the order of embeddable clique minors. 

While we have identified the MSC for a virtual representation of the Chimera lattice, we have not discussed methods for solving the subsequent subgraph isomorphism problem. The latter step is necessary for choosing which member(s) of the MSC contains the input graph as a subgraph. Subgraph isomorphism is believed to be NP-complete and we may expect that solving this problem also poses a significant computational task. However, the instance of interest is strictly smaller than the original minor embedding problem. Since members of the MSC are known beforehand, unique properties of those graphs (which may be used to easily reject potential matches) can be precomputed. These properties include: order, size, degree distribution, clique number, and treewidth. Moreover, recent work from Babai \cite{babai2015graph} has shown that graph isomorphism may be computable in quasipolynomial time, a result that would have a profound implications. We defer a more detailed analysis of this step to a subsequent publication.

\section{Conclusions}
Minor embedding of the logical graph describing an input Hamiltonian presents a significant bottleneck in adiabatic quantum programming. Our aim in this work has been to reduce the difficulty of finding an embedding for a known input graph by exploring what graphs can be embedded into a complete bipartite virtual hardware. By defining the MSC, we identify \Kn{nc+1} as the largest clique which is minor embeddable into the \bip{nc}{nc} virtual hardware of the \Ch{n}{n}{c}.

We have developed a general method for constructing the MSC of a fully connected bipartite graph \bip{N}{N}. It was seen that the contraction of edges belonging to a set of $(N-1)$ non-adjacent edges constructs  all members of the MSC. This edge set could be found using a simple greedy algorithm and the method is also applicable to complete bipartite graphs \bip{N}{N^{\prime}} with $N\neq N^{\prime}$. 

For the case of an incomplete bipartite graph, the simple greedy algorithm is of limited use as the number of minors in the MSC can be very large. We focus on the largest clique minor, and determine two criteria that identify graphs which does not have a \Kn{N+1} clique minor: first, any incomplete bipartite graph which does not meet a minimum size $|E| < (N+1)^2 + (N-1)$ and second, any incomplete bipartite graph on which the set of incomplete vertices cannot be covered by $(N-1)$ edges. These two criteria are enough to determine that a graph does not have a \Kn{N+1} clique minor, but they are insufficient to determine if a graph does have a \Kn{N+1} minor. 

By identifying the MSC of an ideal Chimera hardware, we have a solution to the problem of complete graph embedding. If the graph \Kn{nc+1} can be embedded, then any graph of order $n^{\prime}<nc+1$ can be embedded, but the actual vertex map is not known. Determining the robustness of the \Kn{nc+1} clique minor on a Chimera hardware with faulty qubits is still an open ended question. 

\section{Acknowledgements}
This work was supported by the United States Department of Defense and used resources of the Computational Research and Development Programs at Oak Ridge National Laboratory. This manuscript has been authored by UT-Battelle, LLC, under Contract No. DE-AC0500OR22725 with the U.S. Department of Energy. The United States Government retains and the publisher, by accepting the article for publication, acknowledges that the United States Government retains a non-exclusive, paid-up, irrevocable, world-wide license to publish or reproduce the published form of this manuscript, or allow others to do so, for the United States Government purposes. The Department of Energy will provide public access to these results of federally sponsored research in accordance with the DOE Public Access Plan.

\bibliography{graph_minors,minor_embedding,aqc}

\end{document}